\documentclass[a4paper,UKenglish]{lipics-v2019}

\usepackage{microtype}

\newcommand{\eps}{\varepsilon}
\usepackage{graphicx,amssymb,amsmath}
\usepackage{floatrow}
\usepackage{float}
\usepackage{amsmath}
\usepackage{amsthm}
\usepackage{amssymb}
\usepackage{amsfonts}
\usepackage{color}
\usepackage{algorithm}
\usepackage{caption}
\usepackage{MnSymbol}
\usepackage{multirow}
\usepackage{hhline}

\newcommand{\F}{{\cal F}}

\newcommand{\cD}{{\cal D}}
\newcommand{\cI}{{\cal I}}

\newcommand{\ceil}[1]{\lceil#1\rceil}
\newcommand{\floor}[1]{\lfloor#1\rfloor}

\title{On Approximate Range Mode and Range Selection}

\author{Hicham El-Zein}{Cheriton School of Computer Science, University of Waterloo, Canada}{helzein@uwaterloo.ca}{}{}
\author{Meng He}{Faculty of Computer Science, Dalhousie University, Canada}{mhe@cs.dal.ca}{}{}
\author{J. Ian Munro}{Cheriton School of Computer Science, University of Waterloo, Canada}{imunro@uwaterloo.ca}{}{}
\author{Yakov Nekrich}{Cheriton School of Computer Science, University of Waterloo, Canada}{ynekrich@uwaterloo.ca}{}{}
\author{Bryce Sandlund}{Cheriton School of Computer Science, University of Waterloo, Canada}{bcsandlund@uwaterloo.ca}{}{}
\authorrunning{H. El-Zein, M. He, I. Munro, Y.Nekrich, and B. Sandlund} 
\Copyright{Hicham El-Zein, Meng He, J. Ian Munro, Yakov Nekrich, Bryce Sandlund}


\category{}
\relatedversion{}
\supplement{}
\funding{}

\EventEditors{}
\EventNoEds{3}
\EventLongTitle{}
\EventShortTitle{}
\EventAcronym{}
\EventYear{}
\EventDate{}
\EventLocation{}
\EventLogo{}
\SeriesVolume{}
\ArticleNo{} 
\nolinenumbers 

\begin{document}

\maketitle

\begin{abstract}
For any $\epsilon \in (0,1)$, a $(1+\epsilon)$-approximate range mode query asks for the position of an element whose frequency in the query range is at most a factor $(1+\epsilon)$ smaller than the true mode. For this problem, we design an $O(n/\epsilon)$ bit data structure supporting queries in $O(\lg(1/\epsilon))$ time. This is an encoding data structure which does not require access to the input sequence; we prove the space cost is asymptotically optimal for constant $\epsilon$. Our solution improves the previous best result of Greve et al. (Cell Probe Lower Bounds and Approximations for Range Mode, ICALP'10) by reducing the space cost by a factor of $\lg n$ while achieving the same query time. We also design an $O(n)$-word dynamic data structure that answers queries in $O(\lg n /\lg\lg n)$ time and supports insertions and deletions in $O(\lg n)$ time, for any constant $\epsilon \in (0,1)$. This is the first result on dynamic approximate range mode; it can also be used to obtain the first static data structure for approximate 3-sided range mode queries in two dimensions.

We also consider approximate range selection. For any $\alpha \in (0,1/2)$, an $\alpha$-approximate range selection query asks for the position of an element whose rank in the query range is in $[k - \alpha s, k + \alpha s]$, where $k$ is a rank given by the query and $s$ is the size of the query range. When $\alpha$ is a constant, we design an $O(n)$-bit encoding data structure that can answer queries in constant time and prove this space cost is asymptotically optimal. The previous best result by Krizanc et al. (Range Mode and Range Median Queries on Lists and Trees, Nordic Journal of Computing, 2005) uses $O(n\lg n)$ bits, or $O(n)$ words, to achieve constant approximation for range median only. Thus we not only improve the space cost, but also provide support for any arbitrary $k$ given at query time.
  
\end{abstract}

\section{Introduction}
 
The mode and median of a data set are important statistics, widely used across many disciplines. 
Thus, they are frequently computed in applications for data mining, information retrieval and data analytics. 
The range mode and median problems further aim at speeding up the computation of the mode and median in an arbitrary subrange of the given sequence of elements, 
and thus have been studied extensively~\cite{kms2005,bkpt2005,p2008,pg2009,gs2009,gjlt2010,jl2011,bgjs2011,hmn2011,ChanDLMW14,gn2015,cw2016,ginrr2017,zhms2018}.
In these problems, we preprocess a sequence of elements $c_1, c_2, \ldots, c_n$ to answer queries. 
Given two indices $a$ and $b$ with $1\leq a \leq b \leq n$, a {\em range mode query} asks for a position of the most frequent element in $c_{a..b}$ ($c_{a..b}$ denotes $c_a, \ldots, c_b$), 
while a {\em range median query} asks for the position of the median element in $c_{a..b}$. 
A generalization of range median is the {\em range selection query}, which asks for the position of the $k^\mathrm{th}$ smallest element in $c_{a..b}$ for any given $k$.
Thus a range selection query becomes range median if $k = \lceil (b-a+1)/2 \rceil$. 

Due to the massive amounts of electronic data available, linear space data structures are often preferred by modern applications.
The following are the best solutions to these query problems that use $O(n)$ words of space. 
In static settings, Chan et al.~\cite{ChanDLMW14} showed how to answer a range mode query in $O(\sqrt{n/\lg n})$ time. 
By  proving a conditional lower bound, they also gave strong evidence that, if linear space is required, this query time cannot be improved significantly using purely combinatorial methods with current knowledge. 
When updates to elements are allowed, El-Zein et al.~\cite{zhms2018} showed how to support both range mode queries and updates in $O(n^{2/3})$ time.
For range selection, the solution of Chan and Wilkinson~\cite{cw2016} answers queries in  $O(\lg k / \lg\lg n +1)$ time, matching the lower bound of J{\o}rgensen and Larsen~\cite{jl2011} under the cell probe model.
He et al.~\cite{hmn2011} showed how to support range selection in $O((\lg n /\lg\lg n)^2)$ worst-case time and updates in $O((\lg n /\lg\lg n)^2)$ amortized time. 

The query times for range mode in both linear space data structure solutions and conditional lower bounds are much larger than that for many other query problems, including range median.
To provide faster support for queries, researchers have studied {\em approximate range mode}~\cite{gjlt2010}.
To define this query, let $F_x(c_{a..b})$ denote the frequency of an element $x$ in $c_{a..b}$ and $F(c_{a..b})$ denote the frequency of the mode of $c_{a..b}$ ($F(c_{a..b})=\text{max}_xF_x(c_{a..b})$).
Then a $(1+\eps)$-approximate range mode query asks for the position of an element $x$ in $c_{a..b}$  such that $(1 + \eps) \cdot F_x(c_{a..b}) \geq F(c_{a..b})$ for some positive $\eps$. 
This element is called a $(1+\eps)$-approximate mode of $c_{a..b}$. 
Previously, the best result on this problem is that of Greve et al.~\cite{gjlt2010}, which uses $O(n/\eps)$ words of space to support queries in $O(\lg(1/\eps))$ time, for any $\eps \in (0,1)$.

Approximate range median can be defined similarly.
We say that the $i$th smallest element in the query range $c_{a..b}$ has rank $i$. 
Then, for an approximation ratio $\alpha \in (0,1/2)$, an $\alpha$-approximate range median query asks for the position of an element $x$ whose rank in $c_{a..b}$ is between $\lceil s / 2 \rceil - \alpha s$ and $\lceil s / 2 \rceil + \alpha s$, where $s=b-a+1$.
Bose et al.~\cite{bkpt2005} studied this problem, for which they proposed a data structure occupying $O(n/\alpha)$ words of space that answers queries in constant time.
An $\alpha$-approximate range selection query can also be defined, which, for any given $k$, asks for the position of an element $x$ whose rank in $c_{a..b}$ is between $k - \alpha s$ and $k + \alpha s$. 
However, this problem has not been formally studied previously.

To further improve the space efficiency of data structures, researchers have recently studied various query problems in the encoding model~\cite{fh2011,ginrr2017}.
Under this model, a data structure is not allowed to store or assume access to the original data set.
Instead, it should occupy as little space as possible, while providing support for queries.
For example, in this model, Fischer and Heun~\cite{fh2011} studied the range minimum query problem, which asks for the position of the smallest element in $c_{a..b}$. 
They proposed a data structure occupying only $2n + o(n)$ bits with constant query time. 
The range selection problem has also been considered in this model:
Grossi et al.~\cite{ginrr2017} proposed an encoding data structure occupying $O(n\lg \kappa)$ bits for any fixed positive integer $\kappa$, using which a range selection query can be answered in $O(\lg k / \lg\lg n +1)$ time for any $k$ given in the query with $1\le k \le \kappa$.

Naturally, encoding data structures are only relevant when their space occupancy is asymptotically less than the input data, at least for certain choices of parameters.
The space costs of previous results on approximate range mode or median, however, match the size of the input sequence asymptotically when $\eps$ or $\alpha$ is a constant and become superlinear when $\eps$ or $\alpha$ is in $o(1)$. 
Thus, we study the problem of designing encoding data structures of approximate range mode, median and selection queries, to improve the space efficiency of previous solutions.
Furthermore, previously no research has been done on dynamic approximate range mode, while the dynamic exact data structures for range mode require polynomial query and update times. 
Therefore, we also study approximate range mode queries under dynamic settings, to provide substantially faster support for queries and updates. 


\begin{table}
\label{table:rangemode}
\begin{center}
\makebox[0cm]{
\begin{tabular}{|c|p{1.18in}|p{1.18in}|p{1.02in}|c|}
\hline
Query Type & Query Time & Update Time & Space in Bits & Source\\
\hline
\multirow{6}{*}{Exact} & $O(n^\delta \log n)$ & - & $O(n^{2-2\delta}\lg n)$ & \cite{kms2005}\\
\hhline{~----}
& $O(1)$ & - & $O(n^2 \log \log n/\lg n)$ & \cite{pg2009}\\
\hhline{~----}
& $O(\sqrt{n/\log n})$ & - & $O(n\lg n)$ & \cite{ChanDLMW14}\\
\hhline{~----}
& $O(n^{3/4} \log n / \log \log n)$ & $O(n^{3/4}  \log \log n)$ & $O(n\lg n)$ & \cite{ChanDLMW14}\\
\hhline{~----}
& $O(n^{2/3} \log n / \log \log n)$ & $O(n^{2/3} \log n / \log \log n)$ & $O(n^{4/3}\lg n)$ & \cite{ChanDLMW14}\\
\hhline{~----}
& $O(n^{2/3})$ & $O(n^{2/3})$ & $O(n\lg n)$ & \cite{zhms2018}\\
\hline
& $O(\lg\lg n + \lg(1/\eps))$ & - & $O(n\lg n/\eps)$ & \cite{bkpt2005}\\
\hhline{~----}
$(1+\eps)-$ & $O(\lg(1/\eps))$ & - & $O(n\lg n/\eps)$ & \cite{gjlt2010}\\
\hhline{~----}
$\text{Approximation}$ & $O(\lg(1/\eps))$ & - & $O(n /\eps)$ & new \\
\hhline{~----}
& $O(\lg m / \lg\lg m)$ & $O(\lg n/\eps^2)$ & $O(m\lg m)$ & new \\
\hline
\end{tabular}
}
\end{center}
\caption{Static and Dynamic Range Mode Query History. In this table, $\delta$ is an arbitrary constant in $(0, 1/2)$ and $m = \min(n\lg n/\eps,n/\eps^2)$.}
\label{table:rangemode}
\end{table}

\subparagraph*{Our Results.} For $(1+\eps)$-approximate range mode, where $0 < \eps < 1$, we design an encoding data structure using $O(n/\eps)$ bits that can answer a query in  $O(\lg(1/\eps))$ time. 
This is an improvement upon previous best result of Greve et al.~\cite{gjlt2010}, since we match their query time while saving the space cost by a factor of $\lg n$; we assume a word RAM model in which each word has $\Theta(\lg n)$ bits.
We also prove a lower bound to show that any data structure supporting $(1+\eps)$-approximate range mode must use $\Omega(n/(1+\eps))$ bits for any positive $\eps$.
This means that our space cost is asymptotically optimal for constant $\eps$. 
When $\eps$ is not necessarily a constant, as long as $\eps = o(1/\lg n)$, our data structure uses $o(n\lg n)$ bits, i.e., $o(n)$ words, which is asymptotically less than the space needed to encode the original sequence itself. 

For $\alpha$-approximate range selection, where $0 < \alpha < 1/2$, we design encoding data structures for two variants of this problem.
If $k$ is fixed and given in advance, either as a constant or as a function of the size, $s$, of the query range satisfying certain reasonable constraints (e.g., $k = \lceil s/2 \rceil$ for range median), 
we have a solution occupying $O(n/\alpha^2)$ bits that can answer a query in constant time.
If $k$ is not known beforehand and different values of $k$ could be given with each query, we have another encoding structure in $O(n/\alpha^3)$ bits with constant query time.
Our query time matches that of the previous best data structure of Bose et al.~\cite{bkpt2005} which supports range median only, while we decrease the space cost by a factor of $\lg n$ when $\alpha$ is a constant.
As we also show that any approximate range selection data structure must use at least $\Omega(n)$ bits, our data structures are asymptotically optimal for constant $\alpha$.

In dynamic settings, for any $\eps \in (0,1)$, we present an $O(m\lg m)$-bit structure where $m = \min(n\lg n/\eps,n/\eps^2)$.
It supports $(1+\eps)$-approximate range mode in $O(\lg m /\lg\lg m)$ time and insertions/deletions in $O(\lg n /\eps^2)$ time.
When $\eps$ is an arbitrary constant in $(0,1)$, this data structure uses $O(n)$ words, answers queries in  $O(\lg n /\lg \lg n)$ time, and supports updates in $O(\lg n)$ time. 
As the best result on dynamic exact range mode~\cite{zhms2018} requires $O(n^{2/3})$ time for both queries and updates, this approximate solution is much faster for constant $\eps$.
It is also the first result on dynamic approximate range mode. 
Finally, we apply the technique to solve static $(1+\eps)$-approximate three-sided range mode in two dimensions, achieving $O(\lg m)$ time query and occupying $O(m \lg m)$ words of space, where again $m = \min(n\lg n/\eps,n/\eps^2)$. This is another new approximate query problem. 

Tables \ref{table:rangemode} and \ref{table:rangeselect} compare our results to previous work to be surveyed in Section~\ref{previousw}.


\begin{table}
\label{table:rangeselect}
\begin{center}
\makebox[0cm]{
\begin{tabular}{|c|c|c|p{1.15in}|c|}
\hline
Query Type & Query Time & Update Time & Space in Bits & Source\\
\hline
\multirow{6}{*}{Exact} & $O(1)$ & - & $O((n\lg\lg n)^2 / \lg n)$ & \cite{pg2009}\\
\hhline{~----}
& $O(\lg n /\lg\lg n)$ & - & $O(n \log n)$ & \cite{bgjs2011}\\
\hhline{~----}
& $O(\lg k / \lg\lg n +1)$ & - & $O(n\lg n)$ & \cite{cw2016}\\
\hhline{~----}
& $O(\lg^2 n)$ & $O(\lg^2 n)$ & $O(n\lg^2 n)$ & \cite{gs2009}\\
\hhline{~----}
& $O((\lg n /\lg\lg n)^2)$ & $O((\lg n /\lg\lg n)^2)$ & $O(n \lg^2 n / \lg\lg n)$ & \cite{bgjs2011}\\
\hhline{~----}
& $O((\lg n /\lg\lg n)^2)$ & $O((\lg n /\lg\lg n)^2)$ & $O(n\lg n)$ & \cite{hmn2011}\\
\hline
 $\alpha-\text{Approximation}$ & $O(1)$ & - & $O(n\lg n/\alpha)$ & \cite{bkpt2005}\\
\hhline{~----}
 (with fixed $k$) & $O(1)$ & - & $O(n/\alpha^2)$ & new \\
\hline
\multirow{1}{*}{$\alpha-\text{Approximation}$} & $O(1)$ & - & $O( n/\alpha^3)$ & new \\
\hline
\end{tabular}
}
\end{center}
\caption{Static and Dynamic Range Median and Selection Query History.}
\label{table:rangeselect}
\end{table}

\section{Previous Work}
\label{previousw}


\subparagraph*{Range Mode.} Krizanc et al. \cite{kms2005} first studied the static range mode problem and showed that, for any $\delta \in (0, 1/2)$, there is an $O(n^{2-2\delta})$-word solution that answers queries in  $O(n^\delta \log n)$ time.
Setting $\delta = 1/2$ yields an $O(n)$-word data structure supporting range mode in $O(\sqrt{n} \log n)$ time.
They also presented a data structure using $O(n^2 \log \log n / \log n)$ words, or $O(n^2 \log \log n)$ bits, to support queries in constant time.
Chan et al. \cite{ChanDLMW14} further provided a better linear word solution with $O(\sqrt{n/\log n})$ query time.
They also proved a conditional lower bound to show that, with current knowledge, either the query time must be polynomial, or the construction time must be polynomially larger than $n$. 
Later, Greve et al.~\cite{gjlt2010} gave an (unconditional) lower bound in the cell probe model, showing that any structure using $S$ memory cells of $w$-bit words requires $\Omega(\frac{\log n}{\log (Sw/n)})$ time to answer a range mode query.
On the other end of the spectrum, there has been work~\cite{p2008,pg2009} on improving the constant-time query structure of Krizanc et al., and the best solution uses $O(n^2\lg\lg n / \lg^2n)$ words, or $O(n^2\lg\lg n / \lg n)$ bits~\cite{pg2009}.

In dynamic settings, Chan et al. \cite{ChanDLMW14} provided a tradeoff among space cost, query time and update time. This tradeoff implies two important results: using linear space in words, range mode can be supported in $O(n^{3/4} \log n / \log \log n)$ worst-case time while updates can be performed in $O(n^{3/4} \log \log n)$ amortized expected time. Alternatively, they can use $O(n^{4/3})$ words to improve the query and update efficiency to $O(n^{2/3} \log n / \log \log n)$ worst-case time and amortized expected time, respectively.
They also proved a conditional lower bound to show that, with current knowledge, either queries or updates must require polynomial time.
Very recently, El-Zein et al.~\cite{zhms2018} further improved these solutions by designing an $O(n)$-word structure supporting both queries and updates in $O(n^{2/3})$ time.

Bose et al.~\cite{bkpt2005} were the first to study approximate range mode.
They showed how to provide constant-time support for $4$-approximate mode, $3$-approximate mode and $2$-approximate mode using data structures occupying $O(n)$, $O(n\lg\lg n)$ and $O(n\lg n)$ words, respectively. 
For $(1+\eps)$-approximation, they designed an $O(n/\eps)$-word solution that can answer a query in $O(\lg\lg_{1+\eps} n) = O(\lg\lg n + \lg(1/\eps))$ time. 
Greve et al.~\cite{gjlt2010} further improved these results by using $O(n/\eps)$ words of space to support queries in $O(\lg(1/\eps))$ time. 



\subparagraph*{Range Median and Selection.} The study of range median also has a rich history. 
It was also Krizanc et al. \cite{kms2005} who initially proposed this problem.
There have been several solutions with near-quadratic space and constant query time~\cite{kms2005,p2008,pg2009}, the best of which uses $O((n\lg\lg n / \lg n)^2)$ words~\cite{pg2009}. 
For linear-space solutions, following a series of earlier work~\cite{kms2005,gpt2009,gs2009,bj2009}, Brodal et al.~\cite{bgjs2011} first achieved an $O(n)$-word solution that answers range median and selection queries in $O(\lg n /\lg\lg n)$ time.
J{\o}rgensen and Larsen~\cite{jl2011} further improved the query time of range selection to $O(\lg\lg n + \lg k / \lg\lg n)$, where $k$ is the specified query rank.
They also proved that, under the cell probe model, $\Omega(\lg k / \lg\lg n +1)$ time is necessary for any range selection data structure using $O(n\lg^{O(1)}n)$ space.
Chan and Wilkinson~\cite{cw2016} were then the first who designed a linear word solution with $O(\lg k / \lg\lg n +1)$ optimal query time for range selection.
More recently, Grossi et al.~\cite{ginrr2017} proposed an encoding data structure occupying $O(n\lg \kappa)$ bits for any fixed positive integer $\kappa$, using which a range selection query can be answered in $O(\lg k / \lg\lg n +1)$ time for any $k$ given in the query with $1\le k \le \kappa$.
Gawrychowski and Nicholson~\cite{gn2015} presented a space-optimal encoding of range selection which uses even less space, and proved its space cost is optimal within an $o(n)$ additive term in bits, though no support for queries is provided. All of the above results for range selection assume the selection rank $k$ is specified at query time.

In the dynamic case, Gfeller and Sanders~\cite{gs2009} proposed a data structure that uses $O(n\lg n)$ words of space to support range median in $O(\lg^2 n)$ time and insertions and deletions in $O(\lg^2 n)$ amortized time.
The structure of Brodal et al.~\cite{bgjs2011} occupies $O(n \lg n / \lg\lg n)$ words of space, answers queries in $O((\lg n /\lg\lg n)^2)$ worst-case time and supports insertions and deletions in $O((\lg n /\lg\lg n)^2)$ amortized time.
Later He et al.~\cite{hmn2011} improved the space cost to $O(n)$ words while providing the same support for queries and updates. 
The work of Bose et al.~\cite{bkpt2005} is the only work on $\alpha$-approximate range median, for which they proposed a data structure occupying $O(n/\alpha)$ words of space that answers queries in constant time.


\section{Approximate Range Mode}
Before we proceed, we give a few preliminaries. We will at times refer to elements (of $c_{1..n}$ or otherwise) as \emph{colors}. This is because their data type has no significance in frequency applications and thus the term color standardizes the data type.
Furthermore, at times we create indexing such as a value $r_i$ for when the mode in some range $c_{s_i..r_i}$ exceeds a given threshold. It is possible the mode never exceeds such a threshold. To avoid dealing with such corner cases in the rest of this exposition, we make the assumption that our list of elements $c_{1..n}$ is padded at the beginning and end with a sufficient number of one arbitrary color.

We allow non-constant $\eps$. However, in our upper bounds, we make the restriction $\epsilon \leq 1$, to allow simplification in the runtime and space analyses.

\begin{theorem}
\label{theor:mode2}
Any one-dimensional $(1+\eps)$-approximate range mode data structure requires $\Omega(n/(1+\eps))$ bits.
\end{theorem}

We delay the proof of Theorem \ref{theor:mode2} to Appendix \ref{appendix}. We proceed with our new upper bound.

Our data structure consists of two parts. The first part answers {\it low} frequency queries $c_{a..b}$ with $F(c_{a..b}) \leq \lceil 1 / \eps \rceil$, and is exact. The second part answers {\it high} frequency queries $c_{a..b}$ with $F(c_{a..b}) > \lceil 1 / \eps \rceil$, and makes use of the approximation factor.


\subparagraph*{Low Frequencies: $O(n/\eps)$-Bits $O(\lg(1/\eps))$ Query Time. } 
Similar to the data structure of Greve et al.~\cite{gjlt2010},
for $k=0,\ldots, \lceil 1 / \eps \rceil$ let $Q_k$ be an increasing sequence of size $n$ such that $Q_k[i]$ is the largest integer $j \geq i$ satisfying $F(c_{i..j}) = k$.   
Since $Q_k$ is an increasing sequence whose largest element is $n$, 
we store it in $2n + O(n / \lg^2 n)$ bits \cite{Patrascu08} while still accessing its $i^{\mathrm{th}}$ element in constant time\footnote{
We store $Q_k[1]$ and $(Q_k[i]-Q_k[i-1])$ in unary with a $0$ separator between each two consecutive values in a $2n$-bit vector $\psi$ with rank and select structures. 
To access $Q_k[i]$ we count the number of $1$s before the $i^{th}$ $0$ in $\psi$.
}.
The total space used is $O(n/\eps)$ bits.
Given a query range $c_{a..b}$, $F(c_{a..b}) > k$ iff $b > Q_k[a]$.  
Thus, using binary search, 
we can determine if $F(c_{a..b}) < 1/\eps$ and $K=F(c_{a..b})$ in that case.
If $F(c_{a..b}) < 1/\eps$ we return index $Q_{(K-1)}[a] + 1$; otherwise we query the high frequency structure. 
The total time is $O(\lg(1/\eps))$.

\label{sec:encoding}

\subparagraph*{High Frequencies: $O(n/\eps)$-Bits $O(\lg \lg n + \lg(1/\eps))$ Query Time. }
We first present an $O(n/\eps)$-bit structure that answers high frequency $(1+\eps)$-approximate range mode queries in $O(\lg \lg n + \lg(1/\eps))$ time. We start by developing a tool to binary search the frequency of the mode, with the goal of locating a $(1+\eps)$-approximate mode.

\begin{lemma}
\label{lemma:encoding}
There exists a data structure using $O(k\cdot \eps\cdot n/(1+\eps)^k + n/\lg^2 n)$ bits that can find in constant time, for any query range $c_{a..b}$, one of the following that holds:
\begin{enumerate}
  \item $F(c_{a..b}) < (1+\eps)^k/\eps$,
  \item $F(c_{a..b}) > (1+\eps)^k/\eps$, or
  \item $((1+\eps)^{k-1/2})/\eps < F(c_{a..b}) < ((1+\eps)^{k+1/2})/\eps$.
\end{enumerate}
In case 2, we find an element with frequency greater than $(1+\eps)^k/\eps$ in range $c_{a..b}$. In case 3, 
we find an element with frequency greater than $((1+\eps)^{k-1/2})/\eps$ in range $c_{a..b}$.
\end{lemma}
When this structure is present for all $k$ in range $0, \ldots, \lfloor\lg_{1+\eps} \eps n\rfloor$, 
the above trichotomy is sufficient to binary search for an approximate mode of frequency at least $1/\eps$. 
If we ever land in case 3, 
the encoding gives an approximate mode, 
and otherwise, 
we find the $k$ satisfying $(1+\eps)^k/\eps < F(c_{a..b}) < (1+\eps)^{k+1}/\eps$, 
which represents case 2 for value $k$ and case 1 for value $k+1$. 
Since case 2 provides an element with frequency greater than $(1+\eps)^k/\eps$, this element is an approximate mode.
\begin{proof}
Let $1+\Delta = \sqrt{1+\eps}$ and $f_{j}=(\Delta/\eps)\cdot (1+\Delta)^{j}$.
For each integer $i$ in $[0, n/\lceil f_{2k-1}\rceil]$ let $s_i=i\cdot \lceil f_{2k-1} \rceil + 1$ and denote by $r_i$ the smallest value such that $F(c_{s_i..r_i})\geq (1+\Delta)^{2k}/\eps$.
Notice that $c_{r_i}$ is the unique mode of $c_{s_i..r_i}$. 
Similarly, for each integer $i$ in $[0, n/\lceil f_{2k}\rceil]$, 
let $s'_i=i\cdot \lceil f_{2k} \rceil + 1$ and denote by $r'_i$ the smallest value such that $F(c_{s'_i..r'_i})\geq(1+\Delta)^{2k+1}/\eps$. 

Given a query range $c_{a..b}$, we find the biggest indices $s_i,s'_j$ preceding or equal to $a$. We proceed as follows.
\begin{enumerate}
\item If $b < r_i$, then $F(c_{a..b}) \leq F(c_{s_i..r_i-1}) < ((1+\Delta)^{2k}/\eps) = ((1+\eps)^k/\eps)$.
\item If $b \geq r'_j$, then $F_{r'_j}(c_{a..b}) > F_{r'_j}(c_{s'_j..r'_j}) - f_{2k}$, since there are at most $\lceil f_{2k} \rceil - 1 < f_{2k}$ elements between $s'_j$ and $a$. Then: 
\begin{align*}
F_{r'_j}(c_{a..b}) &> F_{r'_j}(c_{s'_j..r'_j}) - f_{2k}
                  \geq ((1+\Delta)^{2k+1}/\eps) - (\Delta/\eps)\cdot(1+\Delta)^{2k} 
                  = (1+\Delta)^{2k}/\eps \\
				  &= (1+\eps)^k/\eps.
\end{align*}
\item Suppose $b \geq r_i$ and $b < r'_j$. Since there are at most $\lceil f_{2k-1} \rceil -1 < f_{2k-1}$ elements between $s_i$ and $a$ and since $b \geq r_i$, we have that
\begin{align*}
F_{r_i}(c_{a..b}) &> F_{r_i}(c_{s_i..r_i}) - f_{2k-1}
\geq ((1+\Delta)^{2k}/\eps) - (\Delta/\eps) \cdot (1+\Delta)^{2k-1}
=((1+\Delta)^{2k-1})/\eps\\
&=((1+\eps)^{k-1/2})/\eps.
\end{align*}
Finally, since $b < r'_j$, then $F(c_{a..b}) \leq F(c_{s'_j..r'_j-1}) < ((1+\Delta)^{2k+1})/\eps = ((1+\eps)^{k+1/2})/\eps$.
\end{enumerate}
To store the values $\{r_i\}$, we construct a bit vector of length $O(n)$ as follows. In the bit vector, there are $n$ $0$s. For each $r_i$, we insert a $1$ bit after the $r_i$th $0$ bit in the bit vector.
Thus $r_i$ is equal to the number of $0$s before the $i$th $1$ bit in the bit vector.
A second bit vector of length $O(n)$ is used to encode the values $\{r_i'\}$ in a similar way. 
We then represent these two bit vectors in the succinct data structure of Patrascu \cite{Patrascu08}. This data structure provides constant time rank and select, which allow us to locate $r_i$ and $r'_j$, and thus determine whether case 1, 2, or 3 applies, in constant time.

For a bit vector of size $n$ with $m$ ones, the space cost can be made $O(m \lg(n/m)  + n/\lg^2 n)$ bits \cite{Patrascu08}. For vector $r$, $m\lg(n/m) = O((n/f_{2k-1}) \lg f_{2k-1})$, and for vector $r'$, $m\lg(n/m) = O((n/f_{2k}) \lg f_{2k})$. The cost is dominated by vector $r$. Let us first consider the $O(m\lg(n/m))$ term. We have
\begin{equation}
\label{eqbound}
n/f_{2k-1} \lg f_{2k-1} = \frac{\eps n}{\Delta(1+\Delta)^{2k-1}} \lg ((\Delta/\eps) \cdot (1+\Delta)^{2k-1}).
\end{equation}
Rationalizing the denominator, 
we can show $\frac{1}{\Delta} = \frac{1}{\sqrt{1+\eps}-1} = \frac{1+\sqrt{1+\eps}}{\eps}$ and so $\frac{1}{\Delta} = \Theta(\frac{1}{\eps})$ and $\Delta/\eps = O(1)$.
Thus, with $\epsilon \leq 1$, we can bound \eqref{eqbound} with $O\left(\frac{(k-1/2)n}{(1+\eps)^{k-1/2}} \lg (1+\eps)\right)$.
Finally, since we restrict $\eps \leq 1$, we can do a Taylor series expansion to give $\lg(1+\eps) = O(\eps)$. Thus our final space bound is $O((n/f_{2k-1}) \lg f_{2k-1} + n/\lg^2 n) = O(k\cdot \eps\cdot n/(1+\eps)^k + n/\lg^2 n)$.
\end{proof}

To make the above lemma useful, we must apply it to all $k$ in range $0, \ldots, \lfloor \lg_{1+\eps} \eps n\rfloor$. We first analyze the total space cost of all the $O(k\cdot \eps\cdot n/(1+\eps)^k)$ terms. 
Summing up these terms, we have $O\left(\sum_{k=1}^{\lfloor\lg_{1+\eps} n\rfloor} (k\cdot \eps \cdot n / (1+\eps)^k) \right) = O\left(n\cdot \eps \sum_{k=1}^\infty (k / (1+\eps)^k)\right) = O(n/\eps)$ bits.
The other term comes out to $O(\lg_{1+\eps}(\eps\cdot n) \cdot n/\lg^2 n) \subseteq O\left(\frac{n}{\lg n \lg(1+\eps)}\right)$ bits. 
Again applying Taylor series for $1/\lg(1+\eps) = O(1/\eps)$ gives $O(n/(\eps \lg n))$ bits. Thus the total space cost is 
$O(n/\eps)$ bits.

The time complexity of the binary search is different from a typical binary search.
The number of values of $k$ in the entire range is $O(\lg_{1+\eps} n)$, 
so the complexity of the binary search is
$O(\lg(\lg_{1+\eps} n)) = O(\lg\left(\frac{\lg n}{\lg(1+\eps)}\right)) = 
O(\lg\left(\frac{\lg n}{\eps}\right)) = O(\lg\lg n + \lg(1/\eps))$.

\begin{lemma}
There exists an $O(n/\eps)$-bit data structure that supports one-dimensional $(1+\eps)$-approximate range mode queries in $O(\lg\lg n + \lg(1/\eps))$ time.
\end{lemma}

\label{sec:constant}

\subparagraph*{High Frequencies: $O(n/\eps)$-Bits $O(\lg(1/\eps))$ Query Time.} 
The bottleneck of the approach described in the previous section is the binary search on $k$.
To speed up queries, we store an additional data structure that uses $O(n)$ bits but returns a 4-approximate range mode.
   
\begin{lemma}
\label{theorem:mode4}
There exists an $O(n)$-bit data structure that supports one-dimensional approximate range mode queries in constant time with approximation factor $4$.
\end{lemma}
\begin{proof}
We assume $n$ is a power of $2$. We construct a network of fusion trees \cite{Fredman90}. At the top level, we store two fusion trees $\F_{n/2, l}$ and $\F_{n/2, r}$. The tree $\F_{n/2, l}$ contains the values $e_1, \ldots, e_{\lg n}$, where $e_j$ is the largest index satisfying $F(c_{e_j..n/2}) = 2^j$. $\F_{n/2, r}$ contains the values $e_1, \ldots, e_{\lg n}$, where $e_j$ is the smallest index satisfying $F(c_{n/2..e_j}) = 2^j$. If a query crosses the middle index $n/2$, we query $\F_{n/2,l}$ to get $p_1$, the smallest value greater than or equal to $a$, and we query $\F_{k,r}$ to get $p_2$, the largest value less than or equal to $b$.
We return $p_1$ if $F(c_{p_1..n/2}) > F(c_{n/2..p_2})$ and $p_2$ otherwise.
Clearly, $p_1$ is a $2$-approximate mode for $c_{a..n/2}$ and $p_2$ is a $2$-approximate mode for $c_{n/2..b}$. The true mode has at least half its occurrences in one of these regions, so the value we return is a $4$-approximate mode for $c_{a..b}$.

If the query does not cross the middle, it falls entirely in one of the two sides. We may therefore repeat our fusion tree scheme in a divide and conquer fashion, recursing on the two halves. Eventually, there will be a level of the recursion that intersects the query.

To analyze the total space used, we use the recurrence $S(n) = 2S(n/2) + O(\lg^2 n)$, which solves to $S(n) = O(n)$ bits.

To analyze the time complexity of the query, observe that the fusion trees on $O(\lg n)$ elements with word size $O(\lg n)$ support the necessary predecessor/ successor queries in constant time. However, we must know which fusion trees to query. This involves finding the level of recursion in which the query range intersects the midpoint. This is equivalent to the highest set bit in the XOR of $a$ and $b$, which can be determined in constant time in the word RAM model. With this information, we can do the necessary arithmetic to find the appropriate fusion trees to query, and thus query takes constant time.
\end{proof}

We now return to the $(1+\eps)$-approximation. To answer a query $c_{a..b}$, we first query the 4-approximation structure of Lemma~\ref{theorem:mode4}, which returns a corresponding frequency $x$. We now know $x \leq F(c_a,\ldots,c_b) \leq 4x$. We have thus shrunk the number of values of $k$ from Lemma~\ref{lemma:encoding} that need be tested for the $(1+\eps)$-approximation from $\lceil\lg_{1+\eps} n\rceil$ to $\lceil\lg_{1+\eps} (4x/x)\rceil = \lceil\lg_{1+\eps}4\rceil$. Thus our binary search now takes time $O(\lg\left(\frac{2}{\lg(1+\eps)}\right)) = O(\lg(1/\eps))$.

\begin{theorem}
There exists an $O(n/\eps)$-bit data structure that supports one-dimensional approximate range mode queries in $O(\lg(1/\eps))$ time with approximation factor $1+\eps$.
\end{theorem}

\section{Dynamic Approximate Range Mode}
In this section we consider the dynamic variant of the approximate range mode problem. 
We maintain our sequence $c_{a..b}$ under insertions and deletions, so that for an arbitrary query range $c_{a..b}$ an approximate range mode can be found efficiently.

The high-level approach is as follows. Similar to Section \ref{sec:encoding}, 
for each $j \leq \lg_{1+\eps} n$, our goal is to maintain a set of intervals $\cI_j$ such that the mode of a query range $c_{a..b}$ occurs more than $(1+\eps)^j$ times if and only if $c_{a..b}$ contains an interval in $\cI_j$.  
Then, for all $j$ and each interval $c_{l..r}$ in $\cI_j$ we maintain the points $(l,r,j)$ in a data structure $\cD$ that supports the following range queries: 
given a query point $(a, b)$, 
return the highest $j$ such that $a \leq l$ and $r \leq b$ for at least one point $(l,r,j)$ in $\cD$. 

However, unlike the sets of intervals maintained in Section \ref{sec:encoding}, 
our construction in this section satisfies the property that a single update affects only a {\it small} number of intervals in the sets $\cI_j$ for all $j$. We now proceed with the technical argument.

Let $S_{x}$ denote the set of positions of the element $x$ in the sequence $c_{1..n}$.
We will denote by $S_{x}[i]$ the position of the $i^\mathrm{th}$ occurrence of element $x$.  
Let $I_x(l,r)$ denote the interval $c_{S_{x}[l]..S_{x}[r]}$. 

Now let $\delta = 1+\eps' = (1+\eps)^{1/3}$ and fix $x$. There are $f= F_x(c_{1..n})$ occurrences of element $x$ in the full range $c_{1..n}$. We will maintain a subset of the $f \choose 2$ possible intervals $I_x (l,r)$ in sets $\cI_{j,x}$, $1\le j \le \ceil{\lg_{\delta} n}$. We will not have need for nested intervals in $\cI_{j,x}$; therefore, we can number each interval of $\cI_{j,x}$ from left to right with $s_k$ the start of interval $k$ and $e_k$ the end of interval $k$, satisfying $s_k \leq s_{k+1}$, $e_k \leq e_{k+1}$. We maintain the following two invariants on the intervals of $\cI_{j,x}$: (1) $\delta^j \leq e_k - s_k \leq \delta^{j+1}$, and (2) $(\eps'/2) \delta^j \leq s_{k+1} - s_k  \leq \eps' \delta^j$, and the number of positions of $S_x$ not covered by an interval of $\cI_{j,x}$ at either end is at most $(\eps'/2) \delta^j$ (so $s_1 \leq (\eps'/2) \delta^j$ and $f - r_{|\cI_{j,x}|} \leq (\eps'/2) \delta^j$).
From our invariants we get the following proposition.
\begin{proposition}
\label{intersectp}
An interval $I_x(l,r)$ intersects at most $2(r-l+1)/(\eps' \delta^j) + O(1/\eps')$ intervals of $\cI_{j,x}$.
\end{proposition}
\begin{proof}
By Invariant (2), we have a gap size of between $(\eps'/2) \delta^j$ and $\eps' \delta^j$ elements between consecutive starting points of intervals of $\cI_{j,x}$. Since each interval has size at most $\delta^{j+1}$, the total number of intervals intersecting $I_x(l,r)$ is at most $2(r-l+1+2\delta^{j+1})/(\eps' \delta^j)$. 
\end{proof}

For each interval $I_x(s_k, e_k)$ of $\cI_{j,x}$, let $\mathtt{pot}(I_x(s_k, e_k)) = e_k - s_k + 1$ denote the number of elements of $S_x$ (and thus positions in the original sequence $c_{1..n}$) that fall between $s_k$ and $e_k$. When we insert or delete an element $x$, by Proposition \ref{intersectp}, we must update the $\mathtt{pot}$ values of $O(1/\eps')$ intervals of $\cI_{j,x}$.  Across all $j$, $1\le j \le \ceil{\log_{\delta} n}$, $O((1/\eps') \lg_\delta n)$ intervals are affected.

During the updates, each affected $\mathtt{pot}(I_x(s_k, e_k))$ value is incremented or decremented by one. If, for an interval $I_x(s_k, e_k)$ in $\cI_{j,x}$, Invariant (1) is violated by the update, then we rearrange the intervals in the neighborhood of $I_x(s_k, e_k)$ as follows. Consider all intervals of $\cI_{j,x}$ that intersect with $I_x(s_k, e_{k+1})$. By Proposition \ref{intersectp}, there are $O(1/\eps')$ such intervals. We remove said intervals and create new intervals in their place with exactly $\ceil{(1+\eps'/2)\delta^j}$ positions of $x$ that fall in each interval. Furthermore, we space them so that Invariant (2) holds when the new intervals are inserted into $\cI_{j,x}$.

To build these intervals, we must be able to efficiently search for elements by rank in $S_x$. As this will not dominate the update cost, we can use a typical order statistic tree, with $O(\lg f) = O(\lg n)$ query and update time. We may construct the new intervals satisfying invariants (1) and (2) with a constant number of queries on $S_x$ per interval, thus in $O((1/\eps') \lg n)$ time overall.

We can analyze the total cost of rebuilds as follows. On each update, we affect $O(1/\eps')$ intervals at each level. However, the affect on $\mathtt{pot}$ is the same for each interval, and when we rebuild, we rebuild a superset of the intervals affected on update. It follows that the total amortized cost of rebuilds is
$\sum_{j=1}^{\ceil{\log_\delta n}} (1/\delta^j) \cdot O((1/\eps') \lg n) = O((\lg n)/\eps'^2)$ per update.

Further, in each update we must update $S_x$ and update the $\mathtt{pot}$ values. These take time $O(\lg n)$ and $O((1/\eps') \lg_\delta n) = O(\lg n / \eps'^2)$, respectively. So far we pay $O(\lg n / \eps'^2)$ per update, but we have yet to describe the data structure that holds each $\cI_{j,x}$, which will also need to be updated during rebuilds.

Consider each interval $I_x(s_k, e_k)$ of $\cI_{j,x}$ as a point $(s_k, e_k, j)$. 
We store each interval of $\cI_{j,x}$, 
across all $1\le j \le \ceil{\lg_{\delta} n}$ and all $x$, 
in a data structure $\cD$ that supports the following range queries: given a query point $(a, b)$, 
return the highest $j$ such that $a \leq l$ and $r \leq b$ for at least one point $(l,r,j)$ in $\cD$. 
Associated with each point, we keep the element $x$ from which it originated.

We first must consider the number of intervals (and thus points) stored in $\cD$. As before, we assume element $x$ occurs $f=F_x(c_{1..n})$ times in $c_{1..n}$. 
Then $|\cI_{j,x}| = O(f/\ceil{\eps' \delta^j})$. 
Across all levels, we can bound the total number of intervals at $O(f \lg_\delta n) = O(f\lg n / \eps')$ or $O(f / \eps'^2)$. 
Accounting for all $x$, the number of intervals in $\cD$ will be $O(m) = O(\min(n\lg n/\eps', n/\eps'^2))$.

\begin{lemma}
\label{lemma:3d3side}
Data structure $\cD$ can be stored in $O(n\lg n)$ bits, where $n$ is the number of elements in $\cD$. Queries and updates can be supported in $O(\lg n/\lg\lg n)$ time.
\end{lemma}
We delay the proof of Lemma \ref{lemma:3d3side} to Appendix \ref{appendix}.

Now suppose we are given a query range $c_{a..b}$. We find the largest $j$ such that some interval from $\cI_{j,x}$ for some $x$ is contained in $c_{a..b}$. Using data structure $\cD$ from Lemma \ref{lemma:3d3side}, we can compute the index $j$ in $O(\lg n/ \lg \lg n)$ time. We return the element $x$ associated with $j$.

\begin{lemma}
The element $x$ returned is a $(1+\eps)$-approximate mode of query range $c_{a..b}$.
\end{lemma}
\begin{proof}
If $c_{a..b}$ contains an interval from $\cI_{j,x}$, then $x$ occurs at least $\delta^j$ times in $c_{a..b}$. On the other hand, we can show that if some $y$ occurs $\delta^{j+3}$ times in $c_{a..b}$, then $c_{a..b}$ contains an interval from $\cI_{j+1,y}$. Recall $\delta^{j+3} = (1+\eps')\delta^{j+2}$. Each interval of $\cI_{j+1,y}$ has size at most $\delta^{j+2}$ and there is a gap of at most $\eps' \delta^j$ elements of $y$ between the start of every interval in $\cI_{j+1,y}$. Then since $\eps' \delta^{j+1} + \delta^{j+2} < (1+\eps')\delta^{j+2}$, it must be that an interval of $\cI_{j+1,y}$ falls in the query range $c_{a..b}$. We therefore know $\delta^j \leq F(c_{a..b}) < \delta^{j+3} = (1+\epsilon)\delta^j$. It follows that $x$ is a $(1+\eps)$-approximate mode of query range $c_{a..b}$.
\end{proof}

This gives us the main theorem for the section.

\begin{theorem}
There exists an $O(m \lg m)$-bit data structure, where $m = \min(n \lg n/\eps, n/\eps^2)$ that answers $(1+\eps)$-approximate range mode queries in $O(\lg m / \lg \lg m)$ time. Insertions and deletions are supported in $O(\lg n/\eps^2)$ time.
\end{theorem}
\begin{proof}
We have $(1+\eps')^3 = (1+\eps)$ and $(1+\eps)^3 = \eps^3 + 3\eps^2 + 3\eps + 1$. The smallest exponent dominates $O(1/\eps')$ since $\eps \leq 1$ and thus $\eps' < \eps \leq 1$. Thus we have $1/\eps' = O(1/\eps)$. As previously stated, the number of intervals in $\cD$ is $O(m)$, where $m = \min(n \lg n/\eps, n/\eps^2)$. The space bound for $\cD$ is thus $O(m \lg m) = \Omega(n \lg n)$ bits, which dominates the total space cost. The query time is $O(\lg m / \lg \lg m)$.

The update cost has four components: Updating $\cD$, updating $S_\alpha$, and updating $\mathtt{pot}$ values for all affected intervals. As previously mentioned, the latter three are dominated by $O(\lg n/\eps'^2) = O(\lg n/\eps^2)$. Via Lemma \ref{lemma:3d3side}, the cost of updating $\cD$ is $O(\lg m/ \lg \lg m)$. Since $m$ is no more than $n/\eps^2$, $\lg m/\lg \lg m$ is dominated by $O(\lg n/\eps^2)$. In total, the cost of updates is $O(\lg n/\eps^2)$.
\end{proof}

We can use our dynamic data structure to obtain a result for approximate range mode queries on two-dimensional points.  
Our data structure can find approximate mode in the case when the query range is bounded on three sides (see Appendix \ref{appendix} for the proof).
\begin{corollary}
  \label{cor:3sid}
There exists a data structure that supports three-sided two-dimensional approximate range mode queries in $O(\log m)$ time and uses $O(m \log m)$ words of $\log n$ bits, where $m = \min(n \lg n/\eps, n/\eps^2)$.
\end{corollary}

\section{Approximate Range Median and Range Selection}
In this section we present solutions to approximate range selection queries. As discussed previously, a range selection query takes two indices $a,b$ of a sequence $c_1, \ldots, c_n$ and must return the index of an element $x$ whose rank in $c_{a..b}$ is between $k - \alpha (b-a+1)$ and $k + \alpha (b-a+1)$. We study two variants. In the first variant, the rank $k$ is supplied prior to construction of the data structure. In the second variant, we allow $k$ to be specified at query time. Here rank is defined so the $i$th smallest element in the range has rank $i$. We also support a specific $k$ depending on the size of the range, i.e. $f(b-a+1) = \ceil{(b-a+1)/2}$, which is the range median problem. We make the restrictions $f(x) \leq f(x+1) \leq f(x)+1$ and $1 \leq f(x) \leq x$.

\begin{theorem}
\label{theor:median2}
Any one-dimensional approximate range median data structure requires $\Omega(n)$ bits.
\end{theorem}
We delay the proof of Theorem \ref{theor:median2} to Appendix \ref{appendix}. 


\subparagraph*{Fixed Rank $f(b-a+1)=k$ Selection.}

We first address the range selection variant with a fixed rank $f(b-a+1)=k$. We use a similar approach to the one in Lemma~\ref{theorem:mode4}. We again assume $n$ is a power of $2$. At the top level, we store values $m_{n/2,i,j}$ ($1 \leq i,j \leq \lceil \lg_{1+\alpha} n \rceil$). Let $r_i = n/2 - (1+\alpha)^i$ and $s_j = n/2+1 + (1+\alpha)^j$. Then $m_{n/2,i,j}$ is the element of rank $f(s_j - r_i + 1)$ in the range $c_{\lfloor r_i\rfloor..\lceil s_j\rceil}$. We then build the structure recursively on the left and right halves of the full range.

Given a query range $c_{a..b}$, we find the appropriate element $m_{t,i,j}$ where $a \leq t$, $t+1 \leq b$, and $i$ and $j$ are largest possible satisfying $a \leq r_i$ and $s_j \leq b$. We return $m_{t,i,j}$.

\begin{lemma}
\label{staticrlemma}
The above data structure returns an $\alpha$-approximate fixed-rank $k$ element of any query range $c_{a..b}$.
\end{lemma}
\begin{proof}
Let $x=r_i-a$ and $y=b-s_j$. Consider the size of $x$. If we let $z = (1+\alpha)^i$, then $x+z < (1+\alpha)z$. It follows $x < \alpha z$. Since $z \leq (t-a+1)$, and applying similarly for $y$, we can show $x < \alpha(t-a+1)$ and $y < \alpha(b-t)$.
The elements in the ranges represented by $x$ and $y$ shift the true rank $k$ element of $c_{a..b}$ at most $x+y < \alpha(b-a+1)$ ranks from $m_{t,i,j}$. It follows that $m_{t,i,j}$ is an $\alpha$-approximate rank $k$ element for range $s_a, \ldots, s_b$.
\end{proof}

As for Theorem~\ref{theorem:mode4}, to find the level to query, we find the highest set bit of $a$ XOR $b$, then find the appropriate index $m_{t,i,j}$ via arithmetic. In total, the query takes constant time.

We now analyze the space required. At the top level, we use $O(\lg^2_{1+\alpha} (n) \cdot \lg n)$ bits, which is equal to $O(\frac{\lg^3 n}{\lg^2(1+\alpha)}) = O(\lg^3 n / \alpha^2)$ bits. Therefore our recurrence is $S(n) = 2S(n/2) + O(\lg^3 n / \alpha^2)$. The recursion tree is leaf-heavy, with total space amounting to $O(n/\alpha^2)$ bits.

\begin{theorem}
There exists an $O(n/\alpha^2)$-bit data structure that supports one-dimensional $\alpha$-approximate fixed-rank $f(b-a+1)=k$ selection queries in constant time. 
\end{theorem}

\subparagraph*{Online Rank $k$ Selection.}

Our data structure from the previous section can be adapted to support queries that specify the rank $k$ at query time. We again assume $n$ is a power of $2$. Let $\delta = 1+\alpha/2$. At the top level we now store values $m_{n/2,i,j,l}$ ($1 \leq i, j, \leq \ceil{\lg_{\delta} n}$, $0 \leq l \leq \floor{1/\alpha}$). Again, we let $r_i = n/2 - \delta^i$ and $s_j = n/2 + 1 + \delta^j$. However, this time, $m_{n/2,i,j,l}$ represents the element of rank $q_l = l\alpha \cdot (s_j-r_i+1)+1$ in  $c_{r_i..s_j}$. As $q_l$ may be fractional, for simplicity we just store both rank $\floor{q_l}$ and $\ceil{q_l}$ elements. We build this structure recursively on both halves of the full range.

Given a query $s_{a..b}$, we again find the appropriate element $m_{t,i,j,l}$ where $a \leq t$, $t+1 \leq b$, $i$ and $j$ are largest possible satisfying $a \leq r_i$ and $s_j \leq b$, and $l$ is chosen so $q_l$ is as close to $k$ as possible. We return $m_{t,i,j,l}$.

\begin{lemma}
The above data structure returns an $\alpha$-approximate rank $k$ element of any query range $c_{a..b}$ and specified rank $k$.
\end{lemma}
\begin{proof}
Again let $x=r_i-a$ and $y=b-s_j$. For the same reasons as in the proof of Lemma \ref{staticrlemma}, we have $x + y < \alpha (b-a+1)/2$.

There are no more than $\alpha \cdot (s_j-r_i+1) \leq \alpha \cdot (b-a+1)$ ranks between each consecutive $q_l$ and $q_{l+1}$. Thus our chosen $q_l$ satisfies $|q_l-k| < \floor{\alpha(b-a+1)/2}$. It follows that $m_{t,i,j,l}$ is no more than $\alpha \cdot (b-a+1)$ ranks away from the true rank $k$ element in range $c_{a..b}$.
\end{proof}

The query time follows as in the previous section. However, we must account for the additional space usage. Our recurrence is now $T(n) = 2T(n/2) + O(\lg^3 n/\alpha^3)$, from the additional $1/\alpha$ factor in the space cost at each level. This totals to $O(n/\alpha^3)$ bits.

\begin{theorem}
There exists an $O(n/\alpha^3)$-bit  data structure that supports one dimensional $\alpha$-approximate online rank $k$ selection queries in constant time. 
\end{theorem}



\bibliography{../ref}

\newpage

\appendix

\section{Omitted Proofs}
\label{appendix}

\begin{proof}[Proof of Theorem \ref{theor:mode2}]
Using a simple proof we show that $\Omega(n/(1+\eps))$ bits are required for any data structure that answers one-dimensional approximate range mode queries. Here we allow arbitrary $\epsilon$.

Given an approximation factor $1 + \eps$,
divide the sequence $S$ of size $n$ into $\lfloor n / (2k) \rfloor$ {\em full blocks} each of size $2k$, where $k=\lceil 1+\eps \rceil+1$, and, if $n$ is not a multiple of $2k$, a non-full block of size $n\bmod 2k$. 
Denote by $t_1, \ldots, t_{k+1}$ $k+1$ arbitrary, distinct colors.
We say that $S$ satisfies property ($*$) if for each full block $b$ in $S$ one of the following two conditions hold:
\begin{itemize}
\item either $b$ consists of $t_1$ repeated $k$ times followed by $t_2,\ldots,t_{k+1}$,
\item or $b$ consists of $t_2,\ldots,t_{k+1}$ followed by $t_1$ repeated $k$ times.
\end{itemize}
Clearly, the number of sequences that satisfy ($*$) is at least $2^{\lfloor n / (2k) \rfloor}$, since there exist $\lfloor n / (2k) \rfloor$ full blocks in a sequence of size $n$ and each of them can have one of two different values.
Moreover, for any two distinct sequences $S_1$ and $S_2$ satisfying ($*$) differing at full block $b$, 
there exists at least one approximate range mode query, 
namely the query that asks for an approximate mode of $b$, 
that will return different values (either a value from the first $k$ position in the block or from the last $k$ positions of the block).
Thus, the information theoretic lower bound for storing an approximate range mode data structure is $\Omega(\lg 2^{\lfloor n / (2k) \rfloor})=\Omega(\lfloor n / (2k) \rfloor)=\Omega(n)$ bits. 
\end{proof}

\begin{proof}[Proof of Lemma \ref{lemma:3d3side}]
Let $P$ denote the set of points to be stored in our data structure. Here we use $\eps > 0$ independently of the rest of the section.
We start by considering the special case when the second coordinate is bounded by $\lg^{\eps}n$, i.e., $r\le \lg^{\eps} n$ for all $(l,r,j)\in P$.
In this case it is sufficient to store $\lg n$ points for every possible value of $b$:
let $\max_{r,j}$ denote the biggest first coordinate of a point $(l',r',j')$ in $P$ with $r'=r$, $j'=j$ ($\max_{r,j}=\max\{\,l'\,|\, (l',r',j')\in P\, \text{ and }  r'=r, j'=j\}$. 
The answer to a query $(a,b)$ is the largest $j$ that satisfies 
$a \leq \max_{r,j}$ for some $r\le b$. 
We keep all values $\max_{r,j}$ such that $P$ contains at least one point $(l, r, j)$ for some $l$, and store them in increasing order.
We group them in blocks of size $\Theta(\lg^{1-\eps} n)$ and we keep a global lookup table of size $o(n)$ bits that allows answering queries within any possible block.

Also, in a local lookup table of size $O(\lg^{3\eps}n)$ bits we store for each block and every possible value of $r$ the index of the block preceding it which maximizes the value of $j$ given $r$. 
We also store a fusion tree on the values $\max_{r,j}$ so that we can compute the rank of $a$ within these values in constant time.
Given a query, we compute in constant time the block which the predecessor of $a$ belongs to and use table lookup on that block and one other block preceding it to get the answer.
Updates also take constant time since the size of individual blocks and the local lookup table fit in a single word.

A general query can be reduced to the above described special case by using a range tree with node degree $\lg^{\eps}n$ that splits the points on the value of their second coordinate.
Although every point is stored in $O(\lg n/\lg \lg n)$ nodes, our data structure uses linear space. Let $P(u)$ denote the set of points stored in a node $u$. We replace the second coordinate of each point $p\in P(u)$ with the index $i$ of the child node $u_i$ such that $p\in P(u_i)$. We keep the above described special case data structure in every node $P(u)$, but we do not store the set $P(u)$ itself. 
A query interval can be fully covered by $O(\lg n / \lg\lg n)$ tree nodes.
We query the data structure in each one of them and return the maximum value $j$ in $O(\lg n/\lg\lg n)$ time.
Similarly, an update affects the special case data structure in $O(\lg n / \lg\lg n)$ nodes and requires $O(\lg n/\lg\lg n)$ time. 

The total space usage is $O(n \log n)$ bits because we spend $O(\min(\log^{2+\eps}n,|P(u)|\lg n)$ bits in each node $u$ of the range tree. To prove this bound, we classify nodes into low and high nodes. Low nodes are the nodes in the lowest $(1+2/\eps)$ levels of the tree and the rest of the nodes are high nodes. We also store the set of points $P(u)$ in every low node $u$.  Thus we spend $O(|P(u)|\lg n)$ bits in every low node, so the total space consumed by all low nodes is $O((1/\eps)n\lg n)$ bits. We spend $O(\lg^{2+\eps}n)=O(|P(u)|)$ bits in every high node because $|P(u)|\ge \lg^{2+\eps}n$. 
Since the total number of points in all  $P(u)$  is $O(n(\lg n/\lg\lg n))$, the total space consumed by high nodes is $O(n(\lg n/\lg \lg n))$ bits.    
\end{proof}

\begin{proof}[Proof of Corollary \ref{cor:3sid}]
Using the technique introduced by Dietz in~\cite{Dietz89a}, 
we can transform a data structure that supports updates in $u(n)$ time and queries in $q(n)$ time into an offline partially persistent data structure that answers queries in $O(q(n)\cdot \log\log n)$ time and uses $O(n\cdot u(n))$ words of space. 
Using sweep line technique, 
we can transform an offline partially persistent data structure for one-dimensional queries into a static data structure for three-sided queries with the same time and space bounds.
\end{proof}

\begin{proof}[Proof of Theorem \ref{theor:median2}]
Assume $n$ is even. Divide the sequence $S$ of size $n$ to $n / 2$ blocks each of size $2$.
We say that $S$ satisfies property ($*$) if for each block $b$ in $S$ one of the following two conditions hold:
\begin{itemize}
\item either $b$ consists of $\{1, 2\}$,
\item or $b$ consists of $\{2, 1\}$.
\end{itemize}
Clearly, the number of sequences that satisfy ($*$) is $2^{(n / 2)}$ since there exists $n / 2$ blocks in a sequence of size $n$ and each block can have one of two different values.
Moreover for any two distinct sequences $S_1$ and $S_2$ satisfying $(*)$ differing at block $b$, the approximate range selection query must be exact on block $b$, and therefore must return different values.
Thus, the information theoretic lower bound for storing an approximate range median data structure is $\Omega(\lg 2^{(n / 2)})=\Omega(n / 2)=\Omega(n)$ bits. 
\end{proof}

{

\end{document}